\documentclass[12pt,letterpaper]{article}
\usepackage{amsmath,amsfonts,amsthm,amssymb}

\theoremstyle{plain}

\numberwithin{equation}{section}
\newtheorem{thm}{Theorem}[section]
\newtheorem{lem}[thm]{Lemma}

\newcounter{cond}

\newcommand{\complex}{{\mathbb C}}
\newcommand{\real}{{\mathbb R}}
\newcommand{\ascript}{{\mathcal A}}
\newcommand{\oscript}{{\mathcal O}}
\newcommand{\rmre}{\mathrm{Re\,}}
\newcommand{\rmtr}{\mathrm{tr\,}}

\newcommand{\ab}[1]{\left|#1\right|}
\newcommand{\brac}[1]{\left\{#1\right\}}
\newcommand{\paren}[1]{\left(#1\right)}
\newcommand{\sqbrac}[1]{\left[#1\right]}
\newcommand{\ket}[1]{{\left|#1\right>}}
\newcommand{\bra}[1]{{\left<#1\right|}}
\newcommand{\doubleab}[1]{\left\|#1\right\|}
\newcommand{\elbows}[1]{{\left\langle#1\right\rangle}}

\begin{document}

\title{THE UNIVERSE AND\\THE QUANTUM COMPUTER}
\author{Stan Gudder\\ Department of Mathematics\\
University of Denver\\ Denver, Colorado 80208\\
sgudder@math.du.edu}
\date{}
\maketitle

\begin{abstract}
It is first pointed out that there is a common mathematical model for the universe and the quantum computer. The former is called the histories approach to quantum mechanics and the latter is called measurement based quantum computation. Although a rigorous concrete model for the universe has not been completed, a quantum measure and integration theory has been developed which may be useful for future progress. In this work we show that the quantum integral is the unique functional satisfying certain basic physical and mathematical principles. Since the set of paths (or trajectories) for a quantum computer is finite, this theory is easier to treat and more developed. We observe that the sum of the quantum measures of the paths is unity and the total interference vanishes. Thus, constructive interference is always balanced by an equal amount of destructive interference. As an example we consider a simplified two-slit experiment
\end{abstract}
\medskip



\section{Introduction}  
Both the universe and a measurement based quantum computer can be modeled as follows
\begin{equation}         
\label{eq11}
W\to U_1\to M_1\to U_2\to M_2\to\cdots\to U_n\to M_n
\end{equation}
In this process, $W$ is the initial state given by a density operator, $U_1$ is a quantum gate
(or propagator from time 0 to $t_1$) given by a unitary operator, $M_1$ is a quantum event given by a projection operator $P_1(A_1)$, $U_2$ is a quantum gate (or propagator from time $t_1$ to $t_2$) given by a unitary operator, $M_2$ is a quantum event given by a projection operator $P_2(A_2)$, $\ldots\,$.

In the universe model, $A_1$ is one of various possible configurations and $P_1(A_1)$ is the quantum event that the universe is in configuration $A_1$ at time $t_1$. In the quantum computer model $A_1$ is a set of possible outcomes for a measurement and $P_1(A_1)$ is the quantum event that one of the outcomes in $A_1$ occurs. In this model, the measurements can be adaptive in the sense that a choice of measurements may depend on the results of previous measurements. The universe model is referred to as the histories approach to quantum mechanics
\cite{dgt08, djs10, hal09, ish94, mocs05, sor941, sor942, sor07}. A school of researchers believe that this approach is the most promising way to unify quantum mechanics and gravitation. The measurement based quantum computer model is equivalent to the quantum circuit computer model but various researchers believe it has superior properties both in theory and for practical implementation \cite{bbdrv09, gespg07, joz09}. It has been suggested that the universe is itself a gigantic quantum computer. This paper shows that there may be relevance to this statement.

The mathematical background for these models consists of a fixed complex Hilbert space $H$ on which the operators $W$, $U_i$ $P_i(A_j)$ act. In accordance with quantum principles, we assume that $P_i$ is a projection-valued measure, $i=1,2,\ldots\,$. In the universe model, $H$ is infinite dimensional, $t$ is continuous and the sets $A_i$ can be infinite. At the present time, this theory is not complete and is not mathematically rigorous
\cite{djs10, ish94, sor942}. Nevertheless, a quantum measure and integration theory has been developed to treat this approach \cite{gudms, gud09, sor941, sor942, sor07}. In this article we show that the quantum integral defined in
\cite{gud09} is the unique functional that satisfies certain basic physical and mathematical principles.

Since the universe is too vast and complicated for us to tackle in detail now, we move on to the study of quantum computers. From another viewpoint, we are treating toy universes that are described by finite dimensional Hilbert spaces and finite sets. This experience may give us the power and confidence to tackle the real universe later. In this work we make some observations that may be useful in developing the structure for a general theory. For example, we show that the sum of the quantum measure of the paths is unity and that the total interference vanishes. Thus, constructive interference is always balanced by an equal amount of destructive interference. We hope that this article encourages an interchange of ideas between the two groups working on measurement based quantum computation and the histories approach to quantum gravity. It is fascinating to contemplate that the very large and the very small may be two aspects of the same mathematical structure.

\section{Quantum Measures and Integrals} 
One of the main studies of the universe is the field of quantum gravity and cosmology. In this field an important role is played by the histories approach to quantum mechanics \cite{dgt08, djs10, hal09, ish94, sor07}. Let $\Omega$ be the set of \textit{paths} (or \textit{histories} or \textit{trajectories}) for a physical system. We assume  that there is a natural
$\sigma$-algebra $\ascript$ of subsets of $\Omega$ corresponding to the physical \textit{events} of the system and that $\brac{\omega}\in\ascript$ for all $\omega\in\Omega$. In this way, $(\Omega ,\ascript )$ becomes a measurable space. A crucial tool in this theory is a \textit{decoherence functional} $D\colon\ascript\times\ascript\to\complex$ where $D(A,B)$ roughly represents the interference amplitude between events $A$ and $B$. Examples of decoherence functionals for the finite case are given in Section~3. In the infinite case, the form of the decoherence functional is not completely clear \cite{djs10}. However, it is still assumed that $D$ exists and that future research will bear this out. It is postulated that $D$ satisfies the following conditions.

\begin{list} {(\arabic{cond})}{\usecounter{cond}
\setlength{\rightmargin}{\leftmargin}}
\item $D(A,B)=\overline{D(B,A)}$ for all $A,B\in\ascript$.
\item $D(A\cup B,C)=D(A,C)+D(B,C)$ for all $A,B,C\in\ascript$ with $A\cap B=\emptyset$.
\item $D(A,A)\ge 0$ for all $A\in\ascript$.
\item $\ab{D(A,B)}^2\le D(A,A)D(B,B)$ for all $A,B\in\ascript$.
\item If $A_i\in\ascript$ with $A_1\subseteq A_2\subseteq\cdots\,$, then $\lim D(A,A)=D(\cup A_i,\cup A_i)$ and if
$B_i\in\ascript$ with $B_1\supseteq B_2\supseteq\cdots\,$, then $\lim D(B_i,B_i)=D(\cap B_i,\cap B_i)$.
\end{list}

If $D$ is the decoherence functional, then $\mu (A)=D(A,A)$ is interpreted as the ``propensity'' that the event $A$ occurs. We refrain from calling $\mu (A)$ the probability of $A$ because $\mu$ does not have the usual additivity and monotonicity properties of a probability. For example, if $A\cap B=\emptyset$, then
\begin{equation*}
\mu (A\cup B)=D(A\cup B,A\cup B)=\mu (A)+\mu (B)+2\rmre D(A,B)
\end{equation*}
Thus, the interference term $\rmre D(A,B)$ prevents the additivity of $\mu$. Even though $\mu$ is not additive in the usual sense, it does satisfy the more general \textit{grade}-2 \textit{additivity} condition
\begin{equation}         
\label{eq21}
\mu (A\cup B\cup C)=\mu (A\cup B)+\mu (A\cup C)+\mu (B\cup C)-\mu (A)-\mu (B)-\mu (C)
\end{equation}
for any mutually disjoint $A,B,C\in\ascript$. Moreover, by Condition~(5), $\mu$ satisfies the following
\textit{continuity conditions}.
\begin{align}         
\label{eq22}
A_1\subseteq A_2\subseteq\cdots&\Rightarrow\lim\mu (A_i)=\mu (\cup A_i)\\
\label{eq23}
B_1\supseteq B_2\supseteq\cdots&\Rightarrow\lim\mu (B_i)=\mu (\cap B_i)
\end{align}
A grade-2 additive map $\mu\colon\ascript\to\real ^+$ satisfying \eqref{eq22} and \eqref{eq23} is called a
$q$-measure \cite{gudms, gud09, sor941}. A $q$-measure of the form $\mu (A)=D(A,A)$ also satisfies the following
\textit{regularity conditions}.
\begin{align}         
\label{eq24}
\mu (A)=0&\Rightarrow\mu (A\cup B)=\mu (B)\hbox{ for all }B\in\ascript\hbox{ with }A\cap B=\emptyset\\
\label{eq25}
\mu (A\cup B)&=0\hbox{ with }A\cap B=\emptyset\Rightarrow\mu (A)=\mu (B)
\end{align}
A $q$-\textit{measure space} is a triple $(\Omega ,\ascript ,\mu)$ where $\Omega ,\ascript )$ is a measurable space and $\mu\colon\ascript\to\real ^+$ is a $q$-measure \cite{gudms, gud09}.

For a $q$-measure space $(\Omega ,\ascript ,\mu )$ if $A,B\in\ascript$ we define the $(A,B)$
\textit {interference term} $I_{A,B}^\mu$ by
\begin{equation*}
I_{A,B}^\mu=\mu (A\cup B)-\mu (A)-\mu (B)-\mu (A\cap B)
\end{equation*}
Of course, $\mu$ is a measure if and only if $I_{A,B}^\mu =0$ for all $A,B\in\ascript$ and $I_{A,B}^\mu$ describes the amount that $\mu$ deviates from being a measure on $A$ and $B$. Since any $q$-measure $\mu$ satisfies
$\mu (\emptyset )=0$, if $A\cap B=\emptyset$ then
\begin{equation}         
\label{eq26}
I_{A,B}^\mu =\mu (A\cup B)-\mu (A)-\mu (B)
\end{equation}

\begin{lem}       
\label{lem21}
Let $(\Omega ,\ascript ,\mu)$ be a $q$-measure space with $A_i\in\ascript$ mutually disjoint, $i=1,\ldots ,n$. We then have
\begin{align}         
\label{eq27}
\mu\paren{\bigcup _{I=1}^nA_i}&=\sum _{i=1}^n\mu (A_i)+\sum _{i<j=1}^nI_{A_i,A_j}^\mu\\
\noalign{\medskip}
\label{eq28}
\mu\paren{\bigcup _{i=1}^nA_i}&-\mu\paren{\bigcup _{i=2}^nA_i}=\mu (A_1)+\sum _{i=2}^nI_{A_1,A_i}^\mu
\end{align}
\end{lem}
\begin{proof}
Applying Theorem~2.2(b) \cite{gudms} and \eqref{eq26} we have
\begin{align*}
\mu\paren{\bigcup _{i=1}^nA_i}&=\sum _{i<j=1}^n\mu (A_i\cup A_j)-(n-2)\sum _{i=1}^n\mu (A_i)\\
&=\sum _{i<j=1}^n\sqbrac{I_{A_i,A_j}^\mu +\mu (A_i)+\mu (A_j)}-(n-2)\sum _{i=1}^n\mu (A_i)\\
&=\sum _{i=1}^n\mu (A_i)+\sum _{i<j=1}^nI_{A_i,A_j}^\mu
\end{align*}
Equation~\eqref{eq28} follows from \eqref{eq27}
\end{proof}

Let $(\Omega ,\ascript ,\mu)$ be a $q$-measure space and let $f\colon\Omega\to\real ^+$ be a measurable function. We define the \textit{quantum integral} of $f$ to be
\begin{equation}         
\label{eq29}
\int fd\mu =\int _0^\infty\mu\paren{\brac{\omega\colon f(\omega )>\lambda}}d\lambda
\end{equation}
where $d\lambda$ is Lebesgue measure on $\real$ \cite{gud09}. If $f\colon\Omega\to\real$ is measurable, we can write $f$ in a canonical way as $f=f^+-f^-$ where $f^+\ge 0$, $f^-\ge 0$ are measurable and $f^+f^-=0$. We then define the quantum integral
\begin{equation}         
\label{eq210}
\int fd\mu =\int f^+d\mu-\int f^-d\mu
\end{equation}
as long as the two terms in \eqref{eq25} are not both $\infty$. If $\mu$ is an ordinary measure (that is, $\mu$ is additive), then $\int fd\mu$ is the usual Legesgue integral \cite{gud09}. The quantum integral need not be linear or monotone. That is, $\int (f+g)d\mu\ne\int fd\mu +\int gd\mu$ and $\int fd\mu\not\le\int gd\mu$ whenever $f\le g$, in general. However, the quantum integral is homogeneous in the sense that $\int\alpha fd\mu =\alpha\int fd\mu$, for all $\alpha\in\real$. If $\int\ab{f}d\mu <\infty$ we say that $f$ is integrable and we denote by $L_1(\Omega ,\mu )$ the set of integrable functions.

If $\mu$ is a measure on $\ascript$, then it is well known that the Lebesgue integral $f\mapsto\int fd\mu$ from
$L_1(\Omega ,\mu )$ to $\real$ is the unique linear functional satisfying $\int\chi _Ad\mu =\mu (A)$ for all 
$A\in\ascript$ where $\chi _A$ is the characteristic function of $A$ and if $f_i,f\in L_1(\Omega ,\mu )$ with
$0\le f_1\le f_2\le\cdots$, $\lim f_i=f$, then $\lim\int f_id\mu =\int fd\mu$.

We now show that the quantum integral is also the unique functional satisfying certain basic principles.

\begin{thm}       
\label{thm22}
If $(\Omega ,\ascript ,\mu )$ is a $q$-measure space, then $F(f)=\int fd\mu$ is the unique functional
$F\colon L_1(\Omega ,\mu )\to\real$ satisfying the following conditions.

\noindent{\rm (i)}\enspace If $0\le\alpha\le\beta$, $\gamma\ge 0$ and $A\cap B =\emptyset$, then
\begin{equation*}
F\sqbrac{\alpha\chi _A+(\beta +\gamma )\chi _B}
  =\alpha\mu (A)+(\beta +\gamma )\mu (B)+u(\alpha ,\beta )I_{A,B}^\mu
\end{equation*}
for some function $u\colon\brac{(\alpha ,\beta )\in\real ^2\colon 0\le\alpha\le\beta}\to\real$.

\noindent{\rm (ii)}\enspace If $f,g,h\in L_1(\Omega ,\mu )$ are nonnegative and have mutually disjoint support, then
\begin{equation*}
F(f+g+h)=F(f+g)+F(f+h)+F(g+h)-F(f)-F(g)-F(h)
\end{equation*}
{\rm (iii)}\enspace If $f_i,f\in L_1(\Omega ,\mu )$ with $0\le f_1\le f_2\le\cdots$, $\lim f_i=f$, then
$\lim F(f_i)=F(f)$.\enspace
{\rm (iv)}\enspace $F(f)=F(f^+)-F(f^-)$.
\end{thm}

Before we present the proof of Theorem~\ref{thm22}, let us interpret the four conditions. Letting $\alpha =1$,
$\beta =\emptyset$ in (i) we obtain $F(\chi _A)=\mu (A)$ which shows that $F$ is an extension of $\mu$. Letting
$\gamma =0$ in (i) we obtain
\begin{equation}         
\label{eq211}
F(\alpha\chi _A+\beta\chi _B)=\alpha\mu (A)+\beta\mu (B)+u(\alpha ,\beta )I_{A,B}^\mu
\end{equation}
which is an extension of \eqref{eq26} for $\mu$. This also shows that $F$ is linear for the simple function
$\alpha\chi _A+\beta\chi _B$ except for an interference term. With $\gamma\ne 0$, we can write (i) as
\begin{equation*}
F\sqbrac{(\alpha\chi _A+\beta\chi _B)+\gamma\chi _B}=F(\alpha\chi _A+\beta\chi _B)+F(\gamma\chi _B)
\end{equation*}
which essentially states that $\gamma\chi _B$ does not interfere with $\alpha\chi _A+\beta\chi _B$. Condition~(ii) is grade-2 additivity for $F$ which is natural to expect for a functional extension of $\mu$. Condition~(iii) is a generalization of the continuity property \eqref{eq22}. Finally, Condition~(iv) is the natural way to extend $F$ from nonnegative functions to arbitrary real-valued functions in $L_1(\Omega ,\mu )$.

\begin{proof}
\enspace (of Theorem~\ref{thm22}). It is shown in \cite{gud09} that $f\mapsto\int fd\mu$ satisfies (i)--(iv) where (iii) is called the $q$-dominated monotone convergence theorem. Conversely, suppose $F$ satisfies (i)--(iv). Applying
\eqref{eq211} with $B=\emptyset$ we have that $F(\alpha\chi _A)=\alpha\mu (A)$ for all $A\in\ascript$. If
$I_{A,B}^\mu =0$ for all $A,B\in\ascript$, then $\mu$ is a measure and both $F(f)$ and $\int fd\mu$ are the Lebesgue integral of $f$. If $I_{A,B}^\mu\ne 0$ for $A,B\in\ascript$ with $A\cap B=\emptyset$ we have by \eqref{eq211} and (i) that
\begin{align*}
\alpha\mu (A)+(\beta +&\gamma)\mu (B)+u(\alpha ,\beta +\gamma )I_{A,B}^\mu
=F\sqbrac{\alpha\chi _A+(\beta +\gamma )\chi _B}\\
&=\alpha\mu (A)+(\beta +\gamma )\mu (B)+u(\alpha ,\beta )I_{A,B}^\mu
\end{align*}
Hence, $u(\alpha ,\beta +\gamma )=u(\alpha ,\beta )$ and it follows that $u(\alpha ,\beta )$ does not depend on
$\beta$. Letting $v(\alpha )=u(\alpha ,\beta )$ for all $0\le\alpha\le\beta$ we have that
\begin{equation*}
F(\alpha\chi _A+\beta\chi _B)=\alpha\mu (A)+\beta\mu (B)+v(\alpha )I_{A,B}^\mu
\end{equation*}
If $I_{A,B}^\mu\ne 0$, then letting $\alpha =\beta$ we obtain for $A\cap B=\emptyset$ that
\begin{align*}
\alpha\sqbrac{\mu (A)+\mu (B)}+\alpha I_{A,B}^\mu&=\alpha\mu (A\cup B)=F(\alpha\chi _{A\cup B})\\
&=F(\alpha\chi _A+\alpha\chi _B)\\
&=\alpha\sqbrac{\mu (A)+\mu (B)}+v(\alpha )I_{A,B}^\mu
\end{align*}
Hence, $v(\alpha )=\alpha$ for every $\alpha\ge 0$. We conclude that
\begin{equation}         
\label{eq212}
F(\alpha\chi _A+\beta\chi _B)=\alpha\mu (A)+\beta\mu (B)+\alpha I_{A,B}^\mu
\end{equation}
It follows from (ii) and induction that if $f_1,\ldots ,f_n$ have mutually disjoint support, then
\begin{equation}         
\label{eq213}
F\paren{\sum f_i}=\sum _{i<j=1}^nF(f_i+f_j)-(n-2)\sum _{i=1}^nF(f_i)
\end{equation}
If $0\le\alpha _1\le\cdots\le\alpha _n$ and $A_i\in\ascript$ are mutually disjoint, it follows from
\eqref{eq212}, \eqref{eq213} that
\begin{align*}
F\paren{\sum _{i=1}^n\alpha _i\chi _{A_i}}
 & =\sum _{i<j=1}^nF(\alpha _i\chi _{A_i}+\alpha _j\chi _{A_j})-(n-1)\sum _{i=1}^n\alpha _i\mu (A_i)\\
 &=\sum _{i<j=1}^n\sqbrac{\alpha _i\mu (A_i)+\alpha _j\mu (A_j+\alpha _iI_{A_i,A_j}^\mu}
 -(n-2)\sum _{i=1}^n\alpha _i\mu (A_i)\\
 &=\sum _{i=1}^n\alpha _i\mu (A_i)+\sum _{i<j=1}^n\alpha _iI_{A_i,A_j}^\mu
\end{align*}
It is shown in \cite{gudms} that
\begin{align*}
\int\paren{\sum _{i=1}^n\alpha _i\chi _{A_i}}d\mu
&=\alpha _1\sqbrac{\mu\paren{\bigcup _{i=1}^nA_i}-\mu\paren{\bigcup _{i=2}^nA_i}}\\
&\quad +\cdots +\alpha _{n-1}\sqbrac{\mu (A_{n-1}\cup A_n)-\mu (A_n)}+\alpha _n\mu (A_n)
\end{align*}
By \eqref{eq28} and similar expressions for the other terms, we conclude that
\begin{equation*}
\int\paren{\sum _{i=1}^n\alpha _i\chi _{A_i}}d\mu =F\paren{\sum _{i=1}^n\alpha _i\chi _{A_i}}
\end{equation*}
Hence, $F(f)=\int fd\mu$ for every nonnegative simple function. It follows from (iii) and the $q$-dominated monotone convergence theorem that $F(f)=\int fd\mu$ for every nonnegative $f\in L_1(\Omega ,\mu )$. By (iv), $F(f)=\int fd\mu$
for every $f\in L_1(\Omega ,\mu )$.
\end{proof}

\section{Quantum Computers} 
This section considers measurement based quantum computers \cite{bbdrv09, gespg07,joz09}. The theory is simpler than that in Section~2 because the Hilbert space $H$ is finite dimensional and the sample space $\Omega$ is finite. As discussed in Section~1, we have $n$ measurements given by projection-valued measures $P_1,\ldots ,P_n$. Let $\oscript _i$ be the set of possible outcomes for measurements $P_i$, $i=1,\ldots ,n$. Then $\oscript _i$ is a finite set with cardinality $\ab{\oscript _i}=m_i$, $i=1,\ldots ,n$, where $m_i\le\dim H$. Using the notation
$P_i(a)=P_i\paren{\brac{a}}$, since $P_i$ is a projection-valued measure, we have that
\begin{equation}         
\label{eq31}
\sum\brac{P_i(a)\colon a\in\oscript _i}=I
\end{equation}
$i=1,\ldots ,n$, where $I$ is the identity operator. Also, if $A,B\subseteq\oscript _i$ with $A\cap B=\emptyset$ it follows that $P_i(A)P_i(B)=0$. For $A_i\subseteq\oscript _i$, $i=1,\ldots ,n$, we call $A_1\times\cdots\times A_n$ a
\textit{homogeneous event} (or \textit{course-grained history}) and for $a_i\in\oscript _i$ we call
\begin{equation*}
\brac{a_1}\times\cdots\times\brac{a_n}=(a_1,\ldots ,a_n)
\end{equation*}
a \textit{path} (or \textit{trajectory} or \textit{fine-grained history}). Let $\Omega$ be the set of all paths and
$\ascript =2^\Omega$ the set of \textit{events}. Then $\ab{\Omega}=m_1\cdots m_n$ and
$\ab{\ascript}=2^{m_1\cdots m_n}$.

In accordance with Section~1, we have an initial state $W$ given by a density operator on $H$ and unitary operators $U_i$ on $H$ describing quantum gates, $i=1,\ldots ,n$. For two paths $\omega =(a_1,\ldots ,a_n)$ and
$\omega '=(b_1,\ldots ,b_n)$, the \textit{decoherence functional} $D(\omega ,\omega ')$ is defined by
\begin{align}         
\label{eq32}
D&(\omega ,\omega ')\\
&=\rmtr\sqbrac{WU_1^*P_1(a_1)U_2^*P_2(a_2)
  \cdots U_n^*P_n(a_n)P_n(b_n)U_n\cdots P_2(b_2)U_2P_1(b_1)U_1}\notag
\end{align}
Notice that $D(\omega ,\omega ')=0$ if $a_n\ne b_n$; that is, the paths don't end at the same point. For
$A,B\in\ascript$, we extend the decoherence functional by bilinearity to get \cite{djs10}
\begin{equation*}
D(A,B)=\sum\brac{D(\omega ,\omega ')\colon\omega\in A,\omega '\in B}
\end{equation*}
Then $D$ satisfies the usual properties (1)--(5) (Section~2) of a decoherence functional. It follows that
$\mu (A)=D(A,A)$ is a $q$-measure on $\ascript$. If $W=\ket{\psi}\bra{\psi}$ is a pure state, we have
\begin{equation}         
\label{eq33}
\mu (\omega )=D(\omega ,\omega )
=\doubleab{P_n(a_n)U_nP_{n-1}(a_{n-1})U_{n-1}\cdots U_2P_1(a_1)U_1\psi}^2
\end{equation}
As in Section~2, we define the $(\omega ,\omega ')$ \textit{interference term} by
\begin{equation*}
I_{\omega ,\omega '}^\mu =\mu\paren{\brac{\omega ,\omega '}}-\mu (\omega )-\mu (\omega ')
  =2\rmre D(\omega ,\omega ')
\end{equation*}
We then have that
\begin{align}         
\label{eq34}
\mu (A)&=\sum\brac{D(\omega ,\omega ')\colon\omega ,\omega '\in A}\notag\\
  &=\sum\brac{\mu (\omega )\colon\omega\in A}
  +\sum\brac{2\rmre D(\omega ,\omega ')\colon\brac{\omega ,\omega '}\subseteq A}\notag\\
  &=\sum\brac{\mu (\omega )\colon\omega\in A}
  +\sum\brac{I_{\omega ,\omega '}^\mu\colon\brac{\omega ,\omega '}\subseteq A}
\end{align}
Although the next result is well known, we include the proof because it is particularly simple in this case.

\begin{thm}       
\label{thm31}
The $q$-measure $\mu$ satisfies the regularity conditions \eqref{eq24}, \eqref{eq25}.
\end{thm}
\begin{proof}
To prove \eqref{eq24}, suppose that $\mu (A)=0$ and $A\cap B=\emptyset$. Applying Condition~(4) of Section~2 we conclude that $D(A,B)=0$. Hence,
\begin{align*}
\mu (A\cup B)&=\sum\brac{D(\omega ,\omega ')\colon\omega ,\omega '\in A\cup B}\\
&=\mu (A)+\mu (B)+2\rmre D(A,B)=\mu (B)
\end{align*}
To prove \eqref{eq25} suppose that $A\cap B=\emptyset$ and $\mu (A\cup B)=0$. Again, applying Condition~(4) we have that
\begin{align*}
0&=\mu (A\cup B)=\mu (A)+\mu (B)+2\rmre D(A,B)\ge\mu (A)+\mu (B)-2\ab{D(A,B)}\\
&\ge\mu (A)+\mu (B)-2\mu (A)^{\frac{1}{2}}\mu (B)^{\frac{1}{2}}=\sqbrac{\mu (A)^{\frac{1}{2}}-\mu (B)^{\frac{1}{2}}}^2
\end{align*}
Hence, $\mu (A)^{\frac{1}{2}}-\mu (B)^{\frac{1}{2}}=0$ so that $\mu (A)=\mu (B)$.
\end{proof}

Although the next result is elementary, it does not seem to be well known. This result shows that the sum of the quantum measure of the paths is unity and that the total interference vanishes. This indicates that at least one of the outcomes $\omega\in\Omega$ occurs and that constructive interference is always balanced by an equal amount of destructive interference.
\begin{thm}       
\label{thm32}
For the $q$-measure space $(\Omega ,\ascript ,\mu )$ we have
\begin{equation*}
\sum\brac{\mu (\omega )\colon\omega\in\Omega}=1,\quad
  \sum\brac{I_{\omega ,\omega '}^\mu\colon\brac{\omega ,\omega '}\in\ascript}=0
\end{equation*}
\end{thm}
\begin{proof}
We prove this result for a pure state $W=\ket{\psi}\bra{\psi}$ and the general result follows because any state is a convex combination of pure states. Applying \eqref{eq33} we have
\begin{align*}
\sum&\mu (\omega )\\
&=\sum\brac{\doubleab{P_n(a_n)U_nP_{n-1}(a_{n-1})U_{n-1}\cdots U_2P_1(a_1)U_1\psi}^2
\colon (a_1,\ldots ,a_n)\in\Omega}
\end{align*}
By \eqref{eq31} we have that
\begin{align*}
&\sum\mu (\omega )\\
&=\sum\brac{\doubleab{U_nP_{n-1}(a_{n-1})U_{n-1}\cdots U_2P_1(a_1)U_1\psi}^2
\!\colon\!\!(a_2,\ldots ,a_n)\in\oscript _2\times\cdots\times\oscript _n}\\
&\quad\vdots\\
&=\doubleab{U_nU_{n-1}\cdots U_2U_1\psi}^2=1
\end{align*}
Applying \eqref{eq32} we have that
\begin{equation*}
\mu (\Omega )=\sum\brac{D(\omega ,\omega ')\colon\omega ,\omega '\in\Omega}=1
\end{equation*}
Hence, by \eqref{eq34} and what we just proved we obtain
\begin{equation*}
1=\mu (\Omega )=1+\sum\brac{I_{\omega ,\omega '}^\mu\colon\brac{\omega ,\omega '}\in\ascript}
\end{equation*}
and the result follows.
\end{proof}

The next theorem is a straightforward application of \eqref{eq32}

\begin{thm}       
\label{thm33}
{\rm (a)}\enspace If $W=\ket{\psi}\bra{\psi}$ is a pure state and $A_1\times\cdots\times A_n$ is a homogeneous event, then
\begin{equation*}
\mu (A_1\times\cdots\times A_n)=\doubleab{P_n(A_n)U_nP_{n-1}(A_{n-1})\cdots P_1(A_1)U_1\psi}^2
\end{equation*}
{\rm (b)}\enspace If $B_i\in\ascript$ are mutually disjoint, then
\begin{equation*}
\mu\paren{A_1\times\cdots\times A_{n-1}\times\paren{\cup B_i}}
=\sum _i\mu (A_1\times\cdots\times A_{n-1}\times B_i)
\end{equation*}
\end{thm}

The result of Theorem~\ref{thm33}(b) is consistent with the fact that the last measurement does not affect previous ones. We now make two observations. Since
\begin{equation*}
\ab{I_{\omega ,\omega '}^\mu}
=2\ab{\rmre D(\omega ,\omega ')}\le 2\mu (\omega )^{\frac{1}{2}}\mu (\omega ')^{\frac{1}{2}}
\end{equation*}
we have the inequalities
\begin{equation}         
\label{eq35}
\sqbrac{\mu (\omega )^{\frac{1}{2}}-\mu (\omega ')^{\frac{1}{2}}}^2\le\mu\paren{\brac{\omega ,\omega '}}
\le\sqbrac{\mu (\omega )^{\frac{1}{2}}+\mu (\omega ')^{\frac{1}{2}}}^2
\end{equation}
Finally, it is not hard to show that the quantum integral becomes
\begin{align}         
\label{eq36}
\int&fd\mu\\
&=\sum\brac{f(\omega )\mu(\omega )\colon\omega\in\Omega}
+\sum\brac{I_{\omega ,\omega '}^\mu\min\paren{f(\omega ),f(\omega ')}\colon\brac{\omega ,\omega '}\in\ascript}
\notag
\end{align}

We close with an example of a simplified two-slit experiment. Suppose we have two measurements $P_1,P_2$ that have two values $a_1,a_2$ and $b_1,b_2$, respectively. We can assume that $\dim H=2$. We interpret $a_1,a_2$ as two slits and $b_1,b_2$ as two detectors on a detection screen. Suppose we have an initial pure state
$W=\ket{\psi _0}\bra{\psi _0}$. Letting $\psi =U_1\psi _0$ and $U=U_2$ for a homogeneous history $A_1\times A_2$ by Theorem~\ref{thm33}(a) we have that
\begin{equation}         
\label{eq37}
\mu (A_1\times A_2)=\doubleab{P_2(A_2)UP_1(A_1)\psi}^2
\end{equation}
We have the four paths $(a_i,b_j)$, $i,j=1,2$ and \eqref{eq37} gives
\begin{equation*}
\mu\paren{(a_i,b_j)}=\doubleab{P_2(b_j)UP_1(a_i)\psi}^2,\quad i,j=1,2
\end{equation*}
We see directly or by Theorem~\ref{thm32} that
\begin{equation*}
\sum _{i,j=1}^2\mu\paren{(a_i,b_j)}=1
\end{equation*}
so at least one of the paths occurs.

The ``probability'' that detector $b_1$ registers is
\begin{equation*}
\mu\paren{\brac{a_i,a_2}\times\brac{b_1}}=\doubleab{P_2(b_1)U\psi}^2
\end{equation*}
and similarly
\begin{equation*}
\mu\paren{\brac{a_i,a_2}\times\brac{b_2}}=\doubleab{P_2(b_2)U\psi}^2
\end{equation*}
Since
\begin{equation*}
\doubleab{P_2(b_1)U\psi}^2+\doubleab{P_2(b_2)U\psi}^2=1
\end{equation*}
one of the detectors registers. Letting $\omega =(a_1,b_1),\omega '=(a_2,b_1)$, the $(\omega ,\omega ')$ interference term becomes
\begin{equation*}
I_{\omega ,\omega '}^\mu = 2\rmre D(\omega ,\omega ')=2\rmre\elbows{P_1(a_2)U^*P_2(b_1)UP_1(a_1)\psi ,\psi}
\end{equation*}
In general $I_{\omega ,\omega '}^\mu\ne 0$ and hence,
$\mu\paren{\brac{\omega ,\omega '}}\ne\mu (\omega )+\mu (\omega ')$. Thus, the two paths $\omega ,\omega '$ ending at detector $b_1$ interfere. Similarly, the two paths ending at detector $b_2$ interfere.

The ``probability'' that the particle goes through slit $a_1$ is
\begin{equation*}
\mu\paren{\brac{a_1}\times\brac{b_1,b_2}}=\doubleab{UP_1(a_1)\psi}^2=\doubleab{P_1(a_1)\psi}^2
\end{equation*}
and the ``probability'' that the particle goes through slit $a_2$ is
\begin{equation*}
\mu\paren{\brac{a_2}\times\brac{b_1,b_2}}=\doubleab{UP_1(a_2)\psi}^2=\doubleab{P_1(a_2)\psi}^2
\end{equation*}
Since $\doubleab{P_1(a_1)\psi}^2+\doubleab{P_1(a_2)\psi}^2=1$, the particle goes through one of the slits. Again,
\begin{equation*}
\mu\paren{\brac{a_1}\times\brac{b_1,b_2}}=\mu\paren{(a_1,a_2)}+\mu\paren{(a_1,b_2)}
\end{equation*}
so the detectors do not interfere with the slits.

Now $\mu (\emptyset )=0$ and $\mu (\Omega )=1$ so we have $\mu (A)$ for the ten homogeneous events out of the
$\ab{\ascript}=2^4=16$ events in $\ascript$. We have two remaining doubleton sets $\brac{\omega _1,\omega _2}$,
$\brac{\omega _3,\omega _4}$ where $\omega _1=(a_1,b_1)$, $\omega _2=(a_2,b_2)$, $\omega _3=(a_1,b_2)$,
$\omega _4=(a_2,b_1)$. It follows from \eqref{eq32} that in general, any two paths that end at different locations do not interfere. Hence, 
\begin{equation*}
\mu\paren{\brac{\omega _1,\omega _2}}=\mu(\omega _1)+(\omega _2)
\end{equation*}
and similarly
\begin{equation*}
\mu\paren{\brac{\omega _3,\omega _4}}=\mu(\omega _3)+(\omega _4)
\end{equation*}
Finally, we have the four 3-element sets. Applying \eqref{eq34} gives
\begin{align*}
\mu\paren{\brac{(a_1,b_1),(a_1,b_2),(a_2,b_1)}}
  &=\mu\paren{(a_1,b_1)}+\mu\paren{(a_1,b_2)}+\mu\paren{(a_2,b_1)}\\
  &\quad +2\rmre D\paren{(a_1,b_1),(a_2,b_1)}\\
  &=\mu\paren{(a_1,b_1)}+\mu\paren{(a_1,b_2)}+\mu\paren{a_2,b_1)}\\
  &\quad +2\rmre\elbows{P_1(a_2)U^*P_2(b_1)UP_1(a_1)\psi ,\psi}
\end{align*}
In a similar way,
\begin{align*}
\mu\paren{\brac{(a_1,b_1),(a_1,b_2),(a_2,b_2)}}
  &=\mu\paren{(a_1,b_1)}+\mu\paren{(a_1,b_2)}+\mu\paren{(a_2,b_2)}\\
  &\quad +2\rmre D\paren{(a_1,b_2),(a_2,b_2)}\\
  \mu\paren{\brac{(a_1,b_1),(a_2,b_1),(a_2,b_2)}}
  &=\mu\paren{(a_1,b_1)}+\mu\paren{(a_2,b_1)}+\mu\paren{(a_2,b_2)}\\
  &\quad +2\rmre D\paren{(a_1,b_1),(a_2,b_1)}\\
  \mu\paren{\brac{(a_1,b_2),(a_2,b_1),(a_2,b_2)}}
  &=\mu\paren{(a_1,b_2)}+\mu\paren{a_2,b_1)}+\mu\paren{(a_2,b_2)}\\
  &\quad +2\rmre D\paren{(a_1,b_2),(a_2,b_2)}
\end{align*}

To illustrate this example more concretely, we employ the binary notation that is used in the quantum computation literature. We denote the four paths by $00,01,10,11$ and the computational basis for $H=\complex ^2$ by
$\ket{0}$ and $\ket{1}$. The measurements $P_1=P_2$ are relative to the computational basis and are given by
$P(0)=\ket{0}\bra{0}$, $P(1)=\ket{1}\bra{1}$. Suppose $U$ is the Hadamard matrix
\begin{equation*}
U=\frac{1}{\sqrt{2\,}\,}
\left[\begin{matrix}\noalign{\smallskip}1&1\\\noalign{\smallskip}1&-1\\\noalign{\smallskip}\end{matrix}\right]
\end{equation*}
If the initial state is $\psi =\ket{0}$, we have
\begin{align*}
\mu (00)&=\doubleab{\ket{0}\bra{0}U\ket{0}\bra{0}\ket{0}}^2=\tfrac{1}{2}\\
\mu (01)&=\doubleab{\ket{1}\bra{1}U\ket{0}\bra{0}\ket{0}}^2=\tfrac{1}{2}\\
\mu (10)&=\doubleab{\ket{1}\bra{0}U\ket{1}\bra{1}\ket{0}}^2=0\\
\mu (11)&=\doubleab{\ket{1}\bra{1}U\ket{1}\bra{1}\ket{0}}^2=0
\end{align*}
Since $I_{\omega ,\omega '}^\mu=0$ for all $\omega ,\omega '$ there are no interferences. Hence, $\mu$ is a classical measure generated by the previous four measure values.

Next, suppose the initial state is the uniform superposition
\begin{equation*}
\psi =\frac{1}{\sqrt{2\,}\,}\,\ket{0}+\frac{1}{\sqrt{2\,}\,}\,\ket{1}
\end{equation*}
Then
\begin{equation*}
\mu (00)=\doubleab{\,\ket{0}\,\frac{1}{\sqrt{2\,}\,}\,\bra{0}\psi}^2=\frac{1}{4}
\end{equation*}
and similarly, $\mu (01)=\mu (10)=\mu (11)=1/4$. The interference terms become
\begin{align*}
I_{00,01}^\mu&=I_{00,11}^\mu =I_{01,00}^\mu =I_{10,11}^\mu =0\\
I_{00,10}^\mu&=2\rmre\elbows{\ket{1}\bra{1}U^*\ket{0}\bra{0}U\ket{0}\bra{0}\psi ,\psi}=\tfrac{1}{2}\\
I_{01,11}^\mu&=2\rmre\elbows{\ket{1}\bra{1}U^*\ket{1}\bra{1}U\ket{0}\bra{0}\psi ,\psi}=-\tfrac{1}{2}
\end{align*}
The measures of the doubleton events become
\begin{align*}
\mu\paren{\brac{00,01}}
&=\mu\paren{\brac{00,11}}=\mu\paren{\brac{01,10}}=\mu\paren{\brac{10,11}}=\tfrac{1}{2}\\
\mu\paren{\brac{00,10}}=1,\mu\paren{\brac{01,11}}=0
\end{align*}
and the measures of the tripleton events become
\begin{align*}
\mu\paren{\brac{00,01,00}}&=\tfrac{1}{2}+1+\tfrac{1}{2}-\tfrac{3}{4}=\tfrac{5}{4}\\
\mu\paren{\brac{00,01,11}}&=\tfrac{1}{2}+\tfrac{1}{2}+0-\tfrac{3}{4}=\tfrac{1}{4}\\
\mu\paren{\brac{00,10,11}}&=1+\tfrac{1}{2}+\tfrac{1}{2}-\tfrac{3}{4}=\tfrac{5}{4}\\
\mu\paren{\brac{01,10,11}}&=\tfrac{1}{2}+0+\tfrac{1}{2}-\tfrac{3}{4}=\tfrac{1}{4}
\end{align*}
Finally, we illustrate the computation of two quantum integrals for the situation in the last paragraph. Let $f$ be the function that gives the ``length'' of a path where $f(00)=f(11)=1$, $f(01)=f(10)=\sqrt{2\,}$. We then have
\begin{equation*}
\int fd\mu =\tfrac{1}{4}\paren{2+2\sqrt{2\,}\,}+\tfrac{1}{2}-\tfrac{1}{2}=\tfrac{1}{4}\paren{2+2\sqrt{2\,}\,}
\end{equation*}
which is the classical result. To get a nonclassical integral define the function $g$ by $g(00)=0$, $g(01)=g(10)=1$,
$g(11)=2$. We then have
\begin{equation*}
\int gd\mu =\tfrac{1}{4}(0+1+1+2)+0\cdot\tfrac{1}{2}+1\cdot\paren{-\tfrac{1}{2}}=\tfrac{1}{2}
\end{equation*}

\end{document}